\documentclass[journal]{IEEEtran}
\usepackage{graphicx}
\usepackage{subfigure}
\usepackage{epsfig, psfrag, amsmath, amssymb, amsfonts, latexsym, eucal, balance, indentfirst, amsthm}
\usepackage{dsfont}
\usepackage{cite}

\newtheorem{thm}{Theorem}
\newtheorem{cor}{Corollary}
\newtheorem{lem}{Lemma}

\makeatletter

\newcommand{\Rmnum}[1]{\expandafter\@slowromancap\romannumeral #1@}
\makeatother

\graphicspath{{figs/}}

\begin{document}

\sloppy
\title{Dual-Zone Hard-Core Model for RTS/CTS Handshake Analysis in WLANs}
\author{Yi Zhong, \emph{Senior Member, IEEE}, Zhuoling Chen, Wenyi Zhang, \emph{Senior Member, IEEE}, Martin Haenggi, \emph{Fellow, IEEE}
\thanks{
Yi Zhong and Zhuoling Chen are with the School of Electron. Inf. \& Commun., Huazhong University of Science and Technology, Wuhan, China. Wenyi Zhang is with the University of Science and Technology of China. Martin Haenggi is with the University of Notre Dame, US. 
The research has been supported by the National Natural Science Foundation of China (NSFC) grant No. 62471193. Preliminary findings of this work were presented in the IEEE International Conference on Communications (ICC) \cite{6883616}. 

The corresponding author is Yi Zhong (yzhong@hust.edu.cn).}
}

\maketitle

\thispagestyle{plain} 

\begin{abstract}
This paper introduces a new stochastic geometry-based model to analyze the Request-to-Send/Clear-to-Send (RTS/CTS) handshake mechanism in wireless local area networks (WLANs). We develop an advanced hard-core point process model, termed the dual-zone hard-core process (DZHCP), which extends traditional hard-core models to capture the spatial interactions and exclusion effects introduced by the RTS/CTS mechanism. This model integrates key parameters accounting for the thinning effects imposed by RTS/CTS, enabling a refined characterization of active transmitters in the network. Analytical expressions are derived for the intensity of the DZHCP, the mean interference, and an approximation of the success probability, providing insight into how network performance depends on critical design parameters. Our results provide better estimates of interference levels and success probability, which could inform strategies for better interference management and improved performance in future WLAN designs.
\end{abstract}

\section{Introduction}
\label{sec:intro}
\subsection{Motivations}
As a robust complement to the 5G and future 6G mobile communication networks \cite{zhong2024toward}, wireless local area networks (WLANs) will continue to play a pivotal role in data transmission processes. 
The widespread adoption of WLANs in recent years has significantly improved personal access to wireless networks. The IEEE 802.11 media access control (MAC) standard \cite{ieee1997wireless} introduces both physical and virtual carrier sensing techniques to prevent data collisions, a critical issue in densely populated network environments. Physical carrier sensing allows a transmitter to detect active transmissions within its range and to defer its activity if necessary \cite{hasan2007guard}. In contrast, virtual carrier sensing, implemented through the Request-to-Send/Clear-to-Send (RTS/CTS) handshake mechanism, addresses the hidden terminal problem, where a node visible to the receiver is obscured from the transmitter \cite{tobagi1975packet}. This mechanism establishes a protective region around both the transmitter and receiver, effectively minimizing collision risks in overlapping transmission regions (Figure \ref{fig:Pair}).

\begin{figure}[ht]
\centering
\includegraphics[width=0.4\textwidth]{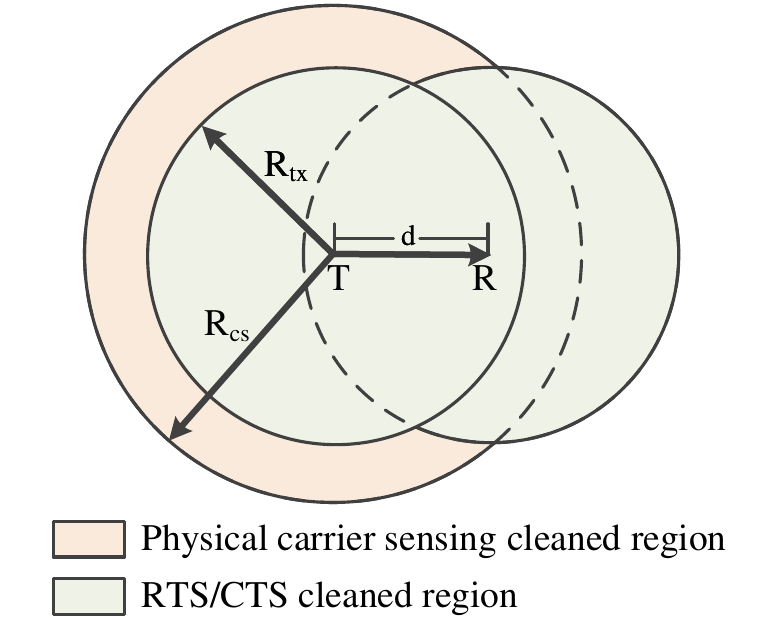}
\caption{Illustration of a transceiver pair employing the RTS/CTS mechanism, showing the dual-zone exclusion regions for interference management. The transmitter (T) and receiver (R) are separated by a distance $d$. 
The orange region, with radius $R_{\mathrm{cs}}$, corresponds to the physical carrier sensing range, which prevents nearby transmitters from causing interference at the transmitter (T). The green region, with radius $R_{\mathrm{tx}}$ and centered at the receiver (R), denotes the RTS/CTS virtual carrier sensing region, which protects the receiver from potential interference by blocking transmissions within this zone.}

\label{fig:Pair}
\end{figure}

Despite considerable research on the performance of 802.11 WLANs \cite{xu2001does, weinmiller1995analyzing, cali1998ieee, bianchi2000performance, xu2003effectiveness}, a precise analysis of the RTS/CTS mechanism's network impact remains challenging. The complexity arises from the randomness in the spatial distribution of nodes and the significant correlations between active transmitters. Understanding the impact of RTS/CTS design parameters, such as carrier sensing ranges on network interference, is crucial for optimizing network performance. 

This paper introduces a novel spatial distribution model specifically designed for analyzing the RTS/CTS mechanism. Traditional models, such as the homogeneous Poisson point process (PPP) and the Matérn hard-core process \cite{haenggi2009stochastic, zhong2020spatio, haenggi2011mean}, fall short due to the unique spatial dependencies introduced by RTS/CTS. Our approach utilizes a marked point process that more accurately characterizes the spatial relationships between transmitters and receivers within this framework. Each node in this process not only indicates the position of a transmitter but also includes a mark representing the receiver's relative location. The influence of RTS/CTS on this model is analogous to the effect observed in Matérn hard-core models on PPPs \cite{matern1986spatial}. Using the proposed model, we derive formulas for calculating the mean interference, accounting for all potential network configurations. Moreover, we employ approximations to derive the success probability of transmissions. A key contribution of this work lies in the application of advanced techniques for stochastic processes, such as reduced Palm expectations and second-order factorial measures, which provide valuable geometric insights into network behavior. Our results highlight the influence of node density and sensing parameters on the effectiveness of the RTS/CTS mechanism.

\subsection{Related Work}
The effective management of network interference and optimization of access protocols in wireless networks, particularly under the RTS/CTS mechanism, has been a focal area of research. 
In \cite{9060334}, the authors investigated the influence of the RTS threshold on the network sum rate, under the condition that the network always operates in the access mode.
The authors of \cite{9369453} propose a new chain RTS/CTS scheme for multihop communication in aerial networks, which achieves 37$\%$ throughput improvement in multihop links with 68$\%$ increased connection establishment time compared to existing IEEE 802.11 Independent Basic Service Set (IBSS).
In \cite{baccelli2012optimizing}, the authors proposed an optimal power control strategy designed to maximize the access probability in carrier sense multiple access (CSMA) protocols by utilizing a max-interference model. While this approach marks a significant advancement in power control tactics, it notably omits a comprehensive statistical analysis of interference, which could otherwise provide deeper insights into the dynamic interplay between power levels and interference patterns across the network.
The authors of \cite{6030792} derive a closed-form interference avoidance criterion based on a multi-frequency RTS/CTS scheme, overlooking the spatial correlation among transmitters \cite{10746553}.

Moreover, traditional interference models for CSMA protocols, such as those based on Matérn's hard-core process discussed in \cite{nguyen2007stochastic, haenggi2011mean, chen2024characterizing}, have been extensively used to analyze network behaviors. These models are appreciated for their theoretical simplicity and tractability. Particularly in \cite{nguyen2007stochastic}, the outage probability was derived using an approximation to the PPP, offering a manageable yet somewhat imprecise methodology for understanding network resilience. Conversely, \cite{haenggi2011mean} advanced this by deriving more accurate expressions for the mean interference within the Matérn hard-core framework, enhancing the model’s utility in practical applications. 
The process has been utilized as a fundamental model to emulate a variety of spatial configurations and scenarios, particularly where entities exhibit a minimum exclusion distance from one another \cite{9266343,9831248,9580432,10319691}.
But the intensity underestimation flaw of Matérn's hard-core process Type \Rmnum{2} has not been addressed in the literature, mentioned in \cite{Busson_Chelius_Gorce_2009} and \cite{6364571}. 
In \cite{6461033}, the authors introduced a modified hard-core point process inspired by the definition of Matérn's hard-core process Type \Rmnum{3} to mitigate the flaw.
Building upon the novel framework presented \cite{6461033}, expressions for outage probability and throughput were derived at the user pair served using NOMA \cite{9531481}.
The modified Matérn hard-core point process Type \Rmnum{3} model, which is an improved simulation method based on Matérn hard-core point process Type \Rmnum{3}, was discussed in order to resolve the inaccurately reflect of modified hard-core point process in \cite{8191026}. 

However, both approaches encounter significant limitations when applied to analyzing networks utilizing the RTS/CTS mechanism. This mechanism introduces additional complexity due to the spatial dependencies among transmitters and their corresponding receivers. In such configurations, the receivers' locations are marked at transmitters and become integral in suppressing potential interference from neighboring transmissions. This intricacy renders traditional tools like the reduced second moment measure and the Ripley K-function, which \cite{haenggi2011mean} employed, less effective. These methods, while robust under simpler point process assumptions, struggle to capture the impact of RTS/CTS protocols on network performance due to their inability to fully integrate the spatial correlation introduced by RTS/CTS interactions.

Recognizing these gaps, our work aims to extend the existing methodologies by incorporating a marked point process approach that explicitly considers the spatial configuration of both transmitters and receivers influenced by RTS/CTS mechanisms. This adaptation allows for a more sophisticated analysis of interference distributions and provides a framework that can more accurately reflect the operational realities of modern WLANs equipped with RTS/CTS protocols. This leads to a more detailed understanding of how spatial arrangements and protocol specifications affect overall network performance, offering valuable insights for network design and optimization.

\subsection{Contributions}
This study proposes an analytical framework for wireless networks based on the RTS/CTS handshake mechanism. Addressing identified gaps in existing research, we introduce novel tools and methodologies that enhance the analysis and optimization of WLAN performance. The key findings and contributions of our study are summarized as follows:

\begin{itemize}
    \item We introduce the dual-zone hard-core process (DZHCP) as a new stochastic model developed for networks utilizing RTS/CTS protocols. This model distinctively captures dual protection regions around transmitters and receivers, enabling a fine-grained analysis of spatial interactions and interference dynamics, thereby enhancing realism in network modeling and interference management.

    \item We have adapted traditional statistical measures, such as the reduced second moment measure and the Ripley K-function, to fit the marked point process models used in our analysis. This adaptation allows these important metrics to effectively analyze complex spatial interactions enabled by the RTS/CTS mechanism, providing a robust framework for characterizing spatial phenomena in wireless networks.

    \item We provide exact expressions for the mean interference based on our new models and measures. These formulas account for the unique spatial dependencies introduced by the RTS/CTS mechanism.

    \item We further contribute by deriving an approximation for the success probability of wireless links under the RTS/CTS mechanism. This approximation, based on asymptotic gain derivations, provides insights into the functioning of real-world networks, enhancing our ability to evaluate and predict network performance under practical operating conditions.
\end{itemize}

The remainder of this paper is organized as follows: Section \ref{sec:model} outlines our marked point process models for nodes under the RTS/CTS protocols, establishing the theoretical framework. Section \ref{sec:analysis} presents our main analytical results, including the intensity of the marked point process and the mean interference. Section \ref{sec:succeprob} presents our approximation for the success probability for the proposed DZHCP. Section \ref{sec:numerical} provides numerical illustrations and demonstrates their practical applications in WLAN environments.

\section{Statistical Model for RTS/CTS Networks}
\label{sec:model}

\subsection{Model Assumptions}
\label{subsec:basic}
We consider a densely populated geographic region with potential transmitter-receiver pairs operating within a shared frequency band, specifically focusing on a single subchannel of an 802.11 WLAN. Subchannels are assumed to be randomly and independently allocated, making this model applicable to multi-channel systems.

Each transmitter broadcasts at a common power level, \(P_t\). Channel fading follows a Rayleigh model, with fading coefficients that remain constant within each time slot and are independent across time and space, modeled as exponentially distributed with unit mean. The fading between a transmitter at point \(x\) and a receiver at point \(z\) is represented by \(h_{xz}\).
For path loss, we use a deterministic function \(l(\cdot)\), commonly \(l(r) = A r^{-\alpha}\) where \(\alpha > 2\) is the path-loss exponent. Under this model, the received power at the receiver at \(z\) from a transmitter at \(x\) is \(h_{xz}l(|x - z|)\).

\subsection{Dual-Zone Hard-Core Process Model}
\label{subsec:ctp}

We model the locations of transmitters and receivers using a Poisson bipolar model \cite{baccelli2009stochastic}, where each point in a PPP represents a potential transmitter, paired with a receiver located at a fixed distance \(d\) and a random angle \(\theta\) (see Figure \ref{fig:bipolar}). Although \(d\) is fixed in this basic model, it can be extended by allowing \(d\) to vary as a random variable.

\begin{figure}
\centering
\includegraphics[width=0.4\textwidth]{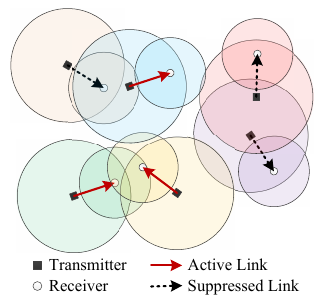}
\caption{Illustration of a bipolar model with RTS/CTS signaling and Type \Rmnum{1} thinning. In this model, each transmitter (represented as a black square) is paired with a corresponding receiver (depicted as a white circle). Active links are shown with solid red arrows, representing the connections that remain functional after the RTS/CTS thinning process. Suppressed links, depicted by dashed black arrows, represent connections that are blocked or ``thinned'' to prevent interference with neighboring active links, in accordance with the RTS/CTS mechanism.}
\label{fig:bipolar}
\end{figure}

Incorporating the RTS/CTS mechanism, each transceiver pair establishes an exclusion region consisting of two regions: a physical carrier sensing region and a virtual carrier sensing region. The physical region is a disk centered on the transmitter with radius \(R_\mathrm{cs}\), while the virtual region includes overlapping disks centered on both the transmitter and receiver, each with radius \(R_\mathrm{tx}\), where \(R_\mathrm{tx} < R_\mathrm{cs}\). 

This dual-zone configuration simulates the protective effect of the RTS/CTS handshake mechanism, using the average received power (ignoring fading) to simplify exclusion region calculations. This approach parallels assumptions used in modeling the CSMA protocol with Matérn hard-core processes. Figure \ref{fig:Pair} illustrates these overlapping regions, highlighting the spatial dynamics governed by the RTS/CTS protocol within our model.

In our network model, we capture the effect of the RTS/CTS handshake mechanism through a process we term \emph{RTS/CTS thinning}. This process selectively filters points from the Poisson bipolar process of potential transceivers, following specific RTS/CTS rules to adapt the network layout and reduce interference. We introduce two types of RTS/CTS thinning, inspired by the Matérn hard-core model but customized to the specifics of wireless protocols using RTS/CTS.

In \emph{Type \Rmnum{1} thinning}, a transceiver pair is retained only if no other potential transmitter exists within its exclusion region. This region is defined by the combined physical and virtual carrier sensing regions. Type \Rmnum{1} thinning models a scenario similar to slotted-time protocols, where each time slot is isolated, and transceivers compete for access within each slot. If a potential collision is detected (i.e., another transmitter is found within the exclusion region), the transceiver pair’s transmission is deferred to avoid interference in the current slot.

\emph{Type \Rmnum{2} thinning} introduces a more dynamic approach by assigning each transceiver pair a random mark, typically viewed as a timestamp. A transceiver pair is retained for transmission only if no other potential transmitter with an earlier timestamp (i.e., a lower mark) is within its exclusion region. This approach more closely represents real-world transmission patterns, where transmissions depend on the relative timing of signals rather than strict time slots. Type \Rmnum{2} thinning thus manages interference by prioritizing established transmissions over newer attempts, creating a transmission landscape that dynamically adjusts to timing variations.

The distinction between these thinning types reflects different strategies for managing network traffic and interference. \emph{Type \Rmnum{1} thinning} provides a straightforward, conservative approach to minimizing signal interference, which can come at the cost of reduced network throughput. By contrast, \emph{Type \Rmnum{2} thinning} introduces greater flexibility by prioritizing transceivers based on their initiation times, potentially enhancing network efficiency. In practical applications, Type \Rmnum{2} thinning aligns more closely with the behavior of modern RTS/CTS protocols, balancing robust interference management with more efficient network utilization.

The proposed dual-zone hard-core processes share some similarities with Matérn hard-core processes—particularly in their generation and thinning methods—but also introduce key differences in the configuration of exclusion regions around each transceiver pair. In a Type \Rmnum{1} Matérn hard-core process, exclusion regions are regular circles with a fixed radius for all transmitters. This symmetry means that if one transmitter falls within the exclusion region of another, the reverse is also true, causing both transmitters to either be retained or removed together.

In contrast, the exclusion regions in a Type \Rmnum{1} dual-zone hard-core process are irregular, resulting in asymmetrical interactions. If one transmitter lies within the exclusion region of another, it does not necessarily mean that the second transmitter will be within the exclusion region of the first. This irregularity allows for scenarios where, under Type \Rmnum{1} thinning in the dual-zone hard-core process, one transceiver pair may be suppressed while the other remains active. This unique configuration introduces a finer control over interference management, enabling more nuanced adjustments based on the spatial relationships between transceivers.

It is worth noting that the dual-zone hard-core model may be particularly suitable for full-duplex communication systems, where simultaneous transmission and reception occur on the same frequency channel. In full-duplex systems, interference management becomes critical, as nodes must carefully handle both self-interference and interference from neighboring nodes \cite{7182305}. The dual exclusion regions—accounting for both physical and virtual carrier sensing—offer a natural framework to model these complex interference patterns, making this model an advantageous choice for analyzing full-duplex networks.
This layered approach to interference exclusion reflects the practical requirements of full-duplex networks more accurately than single-exclusion models.

\subsection{Mathematical Description}
\label{subsec:math-model}

\subsubsection{Type \Rmnum{1} Dual-Zone Hard-Core Process}
The Type \Rmnum{1} dual-zone hard-core process leverages a dependently marked PPP, denoted as $\widetilde{\Phi}_{\rm PPP}$, to represent the spatial distribution and interaction rules for transceivers operating under the RTS/CTS mechanism. This process, defined over the Euclidean plane $\mathbb{R}^2$, uses an intensity $\lambda_p$ to capture the density of potential transmitters and is structured as follows:

\begin{itemize}
    \item {Transmitter locations $\Phi_{\rm PPP}$:} Represents the set of potential transmitters in the network modeled by a PPP, denoted by $\Phi_{\rm PPP} = \{x_i\}$.
    
    \item {Receiver orientation $\theta_i$:} Each transmitter located at $x_i$ has an associated receiver placed at a fixed distance $d$ in a randomly chosen direction. The orientation of each receiver relative to its transmitter is denoted by $\theta_i$, which is uniformly distributed in $[0, 2\pi]$. Together, the pair $(x_i, \theta_i)$ uniquely determines each receiver's location.
    
    \item {Medium access indicator $e_i$:} Each transmitter-receiver pair $(x_i, \theta_i)$ is associated with an access indicator $e_i$, which determines whether the pair is allowed to transmit, based on local network conditions.
\end{itemize}

For a given transmitter-receiver pair $(x_i, \theta_i, e_i)$ in $\widetilde{\Phi}a$, we define its neighborhood $\mathcal{N}(x_i, \theta_i, e_i)$ as
\begin{align}
\mathcal{N}&(x_i,\theta_i,e_i)\triangleq\{(x_j,\theta_j,e_j)\in\widetilde{\Phi}_{\rm PPP}: \nonumber\\
&x_j\in B_{x_i}(R_{\mathrm{cs}})\cup B_{x_i+(d\cos\theta_i,d\sin\theta_i)}(R_{\mathrm{tx}}),j\neq i\},
\end{align}
where $B_x(r)$ denotes a disk of radius $r$ centered at $x$. This neighborhood function $\mathcal{N}$ includes all potential transmitters that fall within the exclusion regions of the transmitter-receiver pair $(x_i, \theta_i, e_i)$, defined by the physical carrier sensing radius $R_{\mathrm{cs}}$ and the virtual carrier sensing radius $R_{\mathrm{tx}}$.

The medium access indicator $e_i$ for a transmitter-receiver pair $(x_i, \theta_i)$ is then determined as
\begin{equation} 
e_i \triangleq \mathds{1}\left(\#\mathcal{N}(x_i, \theta_i, e_i) = 0\right), 
\end{equation} 
where $\#$ denotes the count of elements in the neighborhood. This indicator ensures that a transmitter-receiver pair is allowed to transmit only if no other transmitters are present within its combined exclusion region. In other words, $e_i = 1$ if the exclusion region is empty, allowing transmission, and $e_i = 0$ otherwise.

The collection of all transceiver pairs that successfully meet the conditions of the RTS/CTS thinning process is defined as
\begin{equation}
\widetilde{\Phi} \triangleq \{(x, \theta) : (x, \theta, e) \in \widetilde{\Phi}_{\rm PPP}, \; e = 1\},
\end{equation}
where only pairs with \( e = 1 \) are retained, indicating they have passed the medium access check.
The final set of active transmitters, denoted by \(\Phi\), is
\begin{equation}
\Phi \triangleq \{x : (x, \theta) \in \widetilde{\Phi}\}.
\end{equation}

This filtered set \(\Phi\) represents the transmitters that have successfully passed the RTS/CTS thinning protocol. By restricting the set to these approved transmitters, the model effectively reduces potential interference within the network, allowing for more efficient and reliable communications.

\subsubsection{Type \Rmnum{2} Dual-Zone Hard-Core Process}
In contrast to Type \Rmnum{1}, the Type \Rmnum{2} dual-zone hard-core process incorporates the asynchronous nature of transmission attempts, modeling a more realistic scenario where the RTS/CTS mechanism operates based on chronological transmission order. This process uses a dependently marked PPP, denoted as $\widetilde{\Phi}_{\rm PPP}$, over $\mathbb{R}^2$ with an intensity $\lambda_p$. The components of this process are as follows:

\begin{itemize}
    \item {Transmitter locations $\Phi_{\rm PPP}$:} The set of potential transmitters within the network modeled by a PPP, represented by $\Phi_{\rm PPP} = \{x_i\}$.
    
    \item {Receiver orientation $\theta_i$:} The orientation of each receiver relative to its transmitter at location $x_i$, independently and identically distributed (i.i.d.) over $[0, 2\pi]$. Combined with a fixed transmitter-receiver distance $d$, this orientation determines the exact location of each receiver.
    
    \item {Transmission initiation time marks $\{m_i\}$:} A set of i.i.d. time marks, uniformly distributed over $[0,1]$, representing the relative initiation times of transmissions for each transmitter. These marks are crucial for resolving access contention among transceivers.
    
    \item {Medium access indicator, $e_i$:} Indicates whether a transceiver pair is eligible for transmission based on local network conditions and the timing indicated by $m_i$.
\end{itemize}

For each node $(x_i,\theta_i,m_i,e_i)$ in $\widetilde{\Phi}_{\rm PPP}$, the neighboring set $\mathcal{N}(x_i,\theta_i,m_i,e_i)$ is defined as
\begin{align}
\mathcal{N}&(x_i,\theta_i,m_i,e_i) \triangleq \{(x_j,\theta_j,m_j,e_j) \in \widetilde{\Phi}_{\rm PPP}: \nonumber \\
&x_j \in B_{x_i}(R_{\mathrm{cs}}) \cup B_{x_i + (d\cos\theta_i, d\sin\theta_i)}(R_{\mathrm{tx}}),\ j \neq i\},
\end{align}
where $B_x(r)$ denotes a disk centered at $x$ with radius $r$. This defines the exclusion region for each node, which is influenced by both the physical carrier sensing distance $R_{\mathrm{cs}}$ and the virtual carrier sensing distance $R_{\mathrm{tx}}$.

The eligibility of a node $(x_i,\theta_i,m_i,e_i)$ to transmit, indicated by $e_i$, is
\begin{equation}
e_i \triangleq \mathds{1}\left(\forall (x_j,\theta_j,m_j,e_j) \in \mathcal{N}(x_i,\theta_i,m_i,e_i),\ m_i < m_j \right).
\end{equation}
This condition ensures that a node can transmit only if its time mark $m_i$ is earlier than any other potential transmitters within its exclusion regions.

Following the application of the Type \Rmnum{2} thinning based on these criteria, the set of active transceiver pairs, $\widetilde{\Phi}$, is
\begin{equation}
\widetilde{\Phi} \triangleq \{(x, \theta) : (x, \theta, m, e) \in \widetilde{\Phi}_{\rm PPP}, \ e = 1\}.
\end{equation}
This filtered set $\widetilde{\Phi}$ represents transceiver pairs that successfully passed the Type \Rmnum{2} thinning criteria.

Finally, the collection of active transmitters that survive the thinning process is
\begin{equation}
\Phi \triangleq \{x : (x, \theta) \in \widetilde{\Phi}\}.
\end{equation}
This process effectively captures the asynchronous and competitive dynamics of network access managed by the RTS/CTS mechanism, providing a more realistic, time-sensitive model of network behavior.

\section{Intensity and Mean Interference}
\label{sec:analysis}
In this section, we establish the node intensity and the mean interference experienced by a receiver. We introduce several key quantities that are crucial for our calculations:

\begin{itemize}
    \item {Exclusion region area $V_o$:} This is the area of the exclusion region around a transceiver pair, which is essential for calculating the density of active transmitters:
    \begin{equation}
    V_o = (\pi - \xi_1) R_{\mathrm{cs}}^2 + (\pi - \xi_2) R_{\mathrm{tx}}^2 + d R_{\mathrm{cs}} \sin \xi_1,
    \label{equ:V_o}
    \end{equation}
    where $\xi_1 = \arccos\left(\frac{d^2 + R_{\mathrm{cs}}^2 - R_{\mathrm{tx}}^2}{2 d R_{\mathrm{cs}}}\right)$ and $\xi_2 = \arccos\left(\frac{d^2 + R_{\mathrm{tx}}^2 - R_{\mathrm{cs}}^2}{2 d R_{\mathrm{tx}}}\right)$, as illustrated in Figure \ref{fig:Pair}.

    \item {Spatial criteria for interference $S_1$, $S_2$, and $S_3$:} Define the spatial criteria for transmitters within the exclusion regions, helping to determine interference conditions:
    \begin{itemize}
        \item $S_1 = \{(r, \beta, \theta) : r \leq R_{\mathrm{cs}}\}$,
        \item $S_2 = \{(r, \beta, \theta) : r^2 - 2 r d \cos \beta + d^2 \leq R_{\mathrm{tx}}^2\}$,
        \item $S_3 = \{(r, \beta, \theta) : r^2 + 2 r d \cos(\beta - \theta) + d^2 \leq R_{\mathrm{tx}}^2\}$.
    \end{itemize}
    These sets define how transmitters interact based on their relative distances and orientations, essential for modeling the RTS/CTS mechanism’s impact on network interference. 
    Consider a transceiver pair with its transmitter at the origin \( o \) and receiver at \( (d, 0) \), and another pair with its transmitter at \( (r \cos \beta, r \sin \beta) \) and receiver oriented at \( \theta \). The set \( S_1 \) includes \( (r, \beta, \theta) \) where the second pair's transmitter is within the physical carrier sensing region of the first pair. \( S_2 \) includes \( (r, \beta, \theta) \) where the second pair's transmitter is within the RTS/CTS region influenced by the first pair's receiver. \( S_3 \) covers \( (r, \beta, \theta) \) where the first pair's transmitter is within the RTS/CTS region created by the second pair's transmitter. The union \( S_1 \cup S_2 \cup S_3 \) indicates the conditions under which at least one transceiver pair will be suppressed from transmission.

    \item {Combined sensing area $V(r, \beta, \theta)$:} Describes the combined area of the physical and virtual carrier sensing regions affected by two transceiver pairs separated by distance $r$, with relative phase difference $\beta$, and orientations $0$ and $\theta$, respectively (see Figure \ref{fig:twopair}).
\end{itemize}

\begin{figure}
\centering
\includegraphics[width=0.4\textwidth]{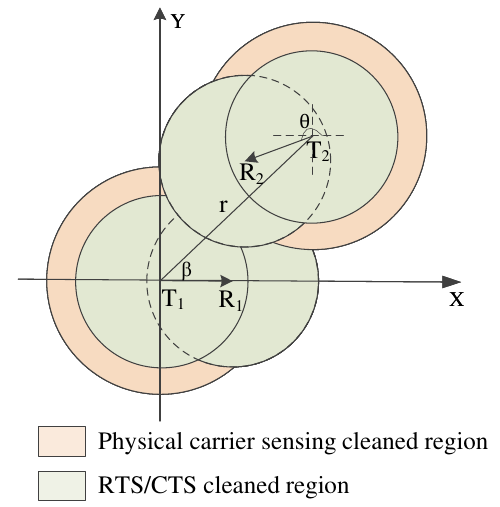}
\caption{Spatial distribution of two transceiver pairs showing their respective exclusion regions.}
\label{fig:twopair}
\end{figure}

These formulations enable us to rigorously calculate the density of active transmitters and the resultant mean interference, accounting for both the geometrical layout and the protective measures instituted by the RTS/CTS mechanism.




\subsection{Type \Rmnum{1} Dual-Zone Hard-Core Process}
\subsubsection{Node Intensity}

In analyzing the node intensity and the mean interference for the Type \Rmnum{1} dual-zone hard-core process, we use the concept of Palm measure. This measure is particularly useful for describing the statistical properties of a point process observed from a given point. For a stationary point process $\Psi$ with a finite, non-zero intensity $\lambda$, and defined on a measurable space $(\mathbb{M}, \mathcal{M})$ with $\mathbb{M}$ representing all possible simple point patterns and $\mathcal{M}$ a $\sigma-$algebra over $\mathbb{M}$, the Palm distribution can be expressed as \cite[Sec. 8.2.1]{haenggi2012stochastic}
\begin{eqnarray}
P_o(Y) = \frac{1}{\lambda v_d(B)} \int_\mathbb{M} \sum_{x \in \psi \cap B} \mathds{1}_Y(\psi_{-x}) P(\mathrm{d}\psi), \quad Y \in \mathcal{M}.
\end{eqnarray}
Here, $v_d$ is the Lebesgue measure, $B$ an arbitrary Borel set with positive measure, and $\psi_{-x} = \{y \in \psi : y - x\}$ denotes the translation of $\psi$ by $-x$.

Focusing on the Type \Rmnum{1} dual-zone hard-core process, where $\Phi$ is a dependently thinned version of the original PPP $\Phi_{\rm PPP}$, we consider the Palm measure $P_o^{(\Phi_{\rm PPP})}$ of $\Phi_{\rm PPP}$. The node intensity for this process is
\begin{eqnarray}
\lambda = \lambda_p P_o^{(\Phi_{\rm PPP})}(o \in \Phi) = \lambda_p e^{-\lambda_p V_o}, \label{equ:lambda_1}
\end{eqnarray}
where $V_o$ is the area of the RTS/CTS exclusion region, factoring in both physical and virtual carrier sensing ranges.

This expression suggests that the node intensity $\lambda$ of the thinned process equals the intensity $\lambda_p$ of the original PPP adjusted by the void probability within the exclusion region defined by $V_o$. This probability highlights how likely a space is free from other nodes, thereby not preventing the considered node's transmission. When $\lambda_p$ is relatively low, $\lambda$ increases linearly with $\lambda_p$. Conversely, as $\lambda_p$ increases beyond a certain threshold, $\lambda$ begins to decrease, underscoring the overly conservative nature of the Type \Rmnum{1} thinning in dense networks. The maximum efficiency, where $\lambda$ peaks, is achieved when $\lambda_p$ is optimally balanced against the exclusion region's area, which occurs at $\lambda_p = 1/V_o$.

\subsubsection{Mean Interference}
This subsection derives the mean interference in the Type \Rmnum{1} dual-zone hard-core process. Without loss of generality, we consider a scenario where the transmitter is located at the origin $o$ and the receiver is at $z_o=(d\cos\theta_o, d\sin\theta_o)$. We specifically condition our analysis on the event $(o,\theta_o) \in \widetilde{\Phi}$, ensuring that the point $(o,\theta_o)$, representing the signal source, does not contribute to the interference calculations. This exclusion is critical as the signal from the considered transceiver pair does not constitute interference to itself.

We employ the \emph{reduced} Palm measure, denoted as $P_{(o,\theta_o)}^!$, associated with the marked point process $\widetilde{\Phi}$ for our interference analysis. The expectation under this measure, denoted by $E_{(o,\theta_o)}^!$, helps calculate the expected interference $I_{z_o}$. Notably, the mean interference calculation is independent of $\theta_o$, allowing us to simplify our notation to $E_{(o,0)}^!(I_{z_o})$. Conditioning on $\theta_o$ provides necessary flexibility for deriving broader results applicable in varying geometrical configurations.

\begin{thm}
\label{thm:inter_1}
The mean interference $I_{z_o}$ that a receiver experiences in the Type \Rmnum{1} dual-zone hard-core process is
\begin{align}
E_{(o,0)}^!&(I_{z_o}) = \frac{\lambda_p^2 P_t}{2\pi \lambda} \int_0^\infty \int_0^{2\pi} \int_0^{2\pi} \nonumber \\
&l(\sqrt{r^2 - 2rd \cos\beta + d^2}) k(r, \beta, 0, \theta) r \mathrm{d}\theta \mathrm{d}\beta \mathrm{d}r, \label{equ:inter_final}
\end{align}
where the kernel function $k(r,\beta,0,\theta)$ is defined as:
\begin{equation}
k(r,\beta,0,\theta) = \left\{
\begin{array}{ll}
0 &\!\!\!\!\!\!\!\!\!\!\!\!\!\!\!\!\!\!\!\!\!\!\!\!\!\!\!\!\!\!\!\!\!\!\!\!\!\!\! \text{if } (r,\beta,\theta) \in S_1 \cup S_2 \cup S_3 \\
\exp(-\lambda_pV(r,\beta,\theta)) & \text{otherwise}.
\end{array}
\right.\label{equ:k_final}
\end{equation}
\end{thm}
\begin{proof}
See Appendix \ref{appendix:A}.
\end{proof}

This theorem encapsulates the core analytical challenge of calculating the interference by integrating over the spatial variables and taking into account the non-uniform distribution of potential interfering transmitters induced by the dual-zone exclusion mechanism. The sets $S_1, S_2,$ and $S_3$ delineate regions where interfering transmitters are either within the physical sensing or the RTS/CTS cleaned regions, thus significantly affecting the resultant interference patterns. By incorporating these spatial dependencies, we offer a comprehensive model that accurately reflects the complexities introduced by the RTS/CTS mechanism in managing interference.

\subsection{Type \Rmnum{2} Dual-Zone Hard-Core Process}
\subsubsection{Node Intensity}
In the Type \Rmnum{2} dual-zone hard-core process, each transceiver pair is marked with a time stamp $t$, which significantly influences the retention of transmitters within the process. For a given transmitter at a specific location, those transmitters within its induced guard region and having a time stamp greater than $t$ are removed, simulating a priority system based on the time of initiation.

This mechanism leads to a situation where the retention probability of a transceiver pair is dependent on its time stamp in relation to others within its vicinity. Specifically, the probability that a given transceiver pair with time stamp $t$ is retained is $e^{-\lambda_p t V_o}$, reflecting the exponential thinning of the process based on the area of the guard region $V_o$ and the density of transmitters $\lambda_p$.

The overall intensity $\lambda$ of the Type \Rmnum{2} process is then derived by integrating this retention probability over all possible values of $t$ from $0$ to $1$, representing a uniform distribution of time stamps:
\begin{align}
\lambda&=\lambda_pP_o^{(\Phi_{\rm PPP})}\left(o\in{\Phi}\right)\nonumber\\
&=\lambda_p\int_0^1P\big((o,\theta_o)\in\widetilde{\Phi}|(o,\theta_o,m_o)\in\widetilde{\Phi}_o,m_o=t\big)\mathrm{d}t\nonumber\\
&=\lambda_p\int_0^1e^{-\lambda_ptV_o}\mathrm{d}t \nonumber \\
&= \frac{1}{V_o}(1-e^{-\lambda_pV_o}), \label{equ:lambda_2}
\end{align}
where the integral simplifies to the expression seen in the final line. Here, $V_o$ is defined as before in (\ref{equ:V_o}) and represents the combined area of both the RTS and CTS regions influencing the potential for a transmitter to be active.

This resulting intensity $\lambda$ provides insight into how the system behaves as $\lambda_p$ increases. Notably, $\lambda$ increases monotonically with $\lambda_p$ and asymptotically approaches the limit $\frac{1}{V_o}$ as $\lambda_p \rightarrow \infty$. This limit signifies that the maximum possible density of active transmitters is inversely proportional to the area of the exclusion region, regardless of the initial density $\lambda_p$ of potential transmitters. Thus, in highly dense environments, the process effectively becomes limited by the spatial constraints imposed by $V_o$, rather than the initial density of transmitters.

\subsubsection{Mean Interference}
Building upon the framework used in the Type \Rmnum{1} process, we derive the mean interference experienced by the typical receiver in the Type \Rmnum{2} dual-zone hard-core process. This process accounts for time stamps that influence transmitter retention, providing a dynamic approach to managing interference.
\begin{thm}
\label{thm:inter_2}
The mean interference $I_{z_o}$ at the typical receiver located at $z_o$ in the Type \Rmnum{2} process is given by
\begin{align}
E_{(o,0)}^!&(I_{z_o})  \frac{\lambda_p^2P_t}{2\pi\lambda}\int_0^\infty\int_0^{2\pi}\int_0^{2\pi} \nonumber \\
&l(\sqrt{r^2-2rd\cos\beta+d^2})k(r,\beta,0,\theta)r\mathrm{d}\theta\mathrm{d}\beta\mathrm{d}r, \label{equ:inter_final2}
\end{align}
where
\begin{equation}
k(r,\beta,0,\theta) = \left\{\begin{array}{ll}
0 & (r,\beta,\theta)\in S_1\bigcup (S_2\bigcap S_3) \\
2\eta(V) &  (r,\beta,\theta)\in \overline{S_1}\bigcap\overline{S_2}\bigcap\overline{S_3}  \\
\eta(V) & \mathrm{otherwise},
\end{array} \right.\label{equ:k_final2}
\end{equation}
in which $V=V(r,\beta,\theta)$, and
\begin{eqnarray}
\eta(V) = \frac{V_oe^{-\lambda_pV}-Ve^{-\lambda_pV_o}+V-V_o}{\lambda_p^2(V-V_o)VV_o}.
\end{eqnarray}
\end{thm}

\begin{proof}
See Appendix \ref{appendix:B}.
\end{proof}

In the theorem, the function $\eta(V)$ accounts for the probability modifications due to the presence of other potential transmitters within these regions, considering both the original and the exclusion region effects. 
The complexity in $k(r, \beta, 0, \theta)$ arises from the necessity to consider overlapping guard regions, where multiple transmitters might influence the retention of a given pair. The conditional structure $\overline{S_1} \cap \overline{S_2} \cap \overline{S_3}$ indicates configurations where neither transmitter lies within the immediate guard region of the other, requiring an adjusted probability function $2\eta(V)$ to account for the doubled effect when both are likely to be retained. The otherwise clause captures all other configurations, applying a single $\eta(V)$ adjustment.

For the case where \(R_{\rm tx} + d < R_{\rm cs}\), meaning the RTS/CTS cleaned region is contained within the physical carrier sensing cleaned region, we have \(S_2 \subset S_1\) and \(S_3 \subset S_1\). In this scenario, the dual-zone hard-core process simplifies to the Matérn hard-core process. Consequently, the function \(k(r,\beta,0,\theta)\) aligns with the results presented in \cite{haenggi2011mean} for both Type \Rmnum{1} and Type \Rmnum{2} Matérn hard-core processes.

\section{Success Probability}
\label{sec:succeprob}
Calculating the success probability \( p_s(\gamma_o) \) of transmissions within spatial point process models requires the probability generating functional (PGFL). However, it is most likely impossible to find an analytical expression for the PGFL of the complex dual-zone hard-core process, which prevents an exact calculation of the success probability.

To address this limitation, prior research has introduced an alternative approach that uses an asymptotic gain to approximate the success probability in cellular networks relative to that of a PPP \cite{6897962, 7322270, 8648502, 9770939}. It was demonstrated that the proposed approximation approach is effective for cellular models, where the intended transmitter is the closest point of a point process. However, its application to bipolar models has not been explored. The interference characteristics in bipolar models, especially under a hard-core constraint, differ substantially due to their more structured point process. 
In this work, we extend the applicability of this approximation method for the first time to the bipolar dual-zone hard-core model, where the distance to the associated transmitter is fixed. 

The proposed method involves first determining the asymptotic gain \( G \) for the point process model \(\phi\). The success probability for \(\phi\), denoted by \( P_\phi(T) \), is then approximated using the success probability of a PPP model with the nearest transmitter association, \( P_{\rm PPP}(T) \), adjusted by \( G \):
\begin{equation}
P_\phi(T) \approx P_{\rm PPP}\left(\frac{T}{G}\right).
\end{equation}

The homogeneous independent Poisson (HIP) model with the nearest transmitter association is used as a reference for deriving \( P_{\rm PPP}(T) \) due to its analytical tractability and simplicity. The HIP model assumes no intra-tier or inter-tier dependencies, and its signal-to-interference ratio (SIR) distribution is independent of node density or transmit power levels. This invariance allows the use of a single-tier framework to compute the Mean Interference-to-Signal Ratio (MISR), which is central to deriving the asymptotic gain.

The asymptotic gain \( G \), as defined in \cite{6897962}, is a measure of the relative SIR characteristics of a PPP and the point process model \(\phi\). It is expressed as the ratio of the MISR of a PPP to that of \(\phi\):
\begin{equation}
G = \frac{{\rm MISR}_{\rm PPP}}{{\rm MISR}_\phi}. \label{eqn:asygain}
\end{equation}

Here, \({\rm MISR}_\phi\) quantifies the average interference relative to the received signal strength in the point process model \(\phi\). For the HIP model with the nearest transmitter association, the MISR is \cite{6897962}
\begin{equation}
{\rm MISR}_{\rm PPP} = \frac{2}{\alpha - 2}.
\end{equation}

We will next provide a detailed exposition on deriving the asymptotic gain \( G \) for both Type \Rmnum{1} and Type \Rmnum{2} dual-zone hard-core processes.

\begin{lem}
\label{lem:asympgain}
The asymptotic gain for Type \Rmnum{1} and Type \Rmnum{2} of the dual-zone hard-core processes is given by
\begin{align}
&G_{\rm DZHCP} = \frac{4\pi\lambda A}{(\alpha-2)\lambda_p^2r_0^\alpha}\bigg(\int_0^\infty\int_0^{2\pi}\int_0^{2\pi} \nonumber\\
&\quad l(\sqrt{r^2-2rd\cos\beta+d^2})k(r,\beta,0,\theta)r\mathrm{d}\theta\mathrm{d}\beta\mathrm{d}r\bigg)^{-1}, \label{eqn:GCTP}
\end{align}
where \( \lambda_p \) is the intensity of the underlying PPP, \(k(r,\beta,0,\theta)\) is given by (\ref{equ:k_final}) for Type \Rmnum{1} of the dual-zone hard-core process and given by (\ref{equ:k_final2})
for Type \Rmnum{2} of the dual-zone hard-core process.
\end{lem}

\begin{proof}
Calculating the MISR for the Type \Rmnum{1} and Type \Rmnum{2} dual-zone hard-core processes is considerably more complex than a PPP model. Unlike the PPP, these models do not benefit from the equivalence properties provided by Slivnyak’s theorem, which simplifies analysis by asserting that conditioning on having a point at \(o\) is equivalent to simply adding \(o\) to the PPP. For the dual-zone hard-core processes, however, we must specifically evaluate the mean interference at a receiver’s location, denoted by \(\mathbb{E}_{(o,0)}^!\left[I_{z_o}\right]\), conditioned on having a transceiver pair with the receiver positioned at that location. 

For the Type \Rmnum{1} and Type \Rmnum{2} processes, the mean interference at the receiver can be calculated using Theorem \ref{thm:inter_1} and Theorem \ref{thm:inter_2}, respectively. The MISR for these models is then determined by dividing the mean interference by the average received power at the reference distance \(r_0\), which is
\begin{align}
&{\rm MISR}_{\rm DZHCP} = \frac{\lambda_p^2r_0^\alpha}{2\pi\lambda A}\int_0^\infty\int_0^{2\pi}\int_0^{2\pi} \nonumber\\
&\qquad l(\sqrt{r^2-2rd\cos\beta+d^2})k(r,\beta,0,\theta)r\mathrm{d}\theta\mathrm{d}\beta\mathrm{d}r, \label{equ:MISR_CTP}
\end{align}
where \(k(r,\beta,0,\theta)\) is defined by (\ref{equ:k_final}) for Type \Rmnum{1} and by (\ref{equ:k_final2}) for Type \Rmnum{2} of the dual-zone hard-core process.

Applying the definition of asymptotic gain from (\ref{eqn:asygain}), we can then derive the corresponding results as outlined in Lemma.
\end{proof}

Considering a Rayleigh fading scenario where the expected value of the channel fading coefficient \(h\) is \(\mathbb{E}[h] = 1\), the success probability for the HIP model with nearest transmitter association for a given SIR threshold \(T\) is \cite{haenggi2009stochastic}
\begin{equation}
P_{\rm PPP}(T) = \left(1 + T^{\frac{2}{\alpha}} \int_{T^{-\frac{2}{\alpha}}}^\infty \frac{1}{1 + t^{\frac{\alpha}{2}}} dt\right)^{-1}.
\end{equation}
This expression encapsulates the cumulative distribution function (ccdf) of the SIR. For a path-loss exponent \(\alpha = 4\), the equation simplifies to
\begin{equation}
P_{\rm PPP}(T) = \frac{1}{1 + \sqrt{T} \arctan(\sqrt{T})}.
\end{equation}

For any point process model \(\phi\) characterized by an asymptotic gain \(G\), as derived in Lemma \ref{lem:asympgain}, the success probability can be approximated by scaling the threshold \(T\) by \(G\):
\begin{equation}
P_\phi(T) \approx P_{\rm PPP}\left(\frac{T}{G}\right) = \left(1 + \left(\frac{T}{G}\right)^{\frac{2}{\alpha}} \int_{\left(\frac{T}{G}\right)^{-\frac{2}{\alpha}}}^\infty \frac{1}{1 + t^{\frac{\alpha}{2}}} dt\right)^{-1}.
\end{equation}
Specifically, for \(\alpha = 4\), the success probability for model \(\phi\) is approximated by
\begin{equation}
P_\phi(T) \approx \left(1 + \sqrt{\frac{T}{G}} \arctan\left(\sqrt{\frac{T}{G}}\right)\right)^{-1}.
\end{equation}

The asymptotic gain \(G\) varies depending on the spatial characteristics and interference patterns specific to the point process model under consideration. By selecting the appropriate asymptotic gain for each model, interference analysis can be precisely tailored to the unique spatial dynamics of each process. The provided approach enables a rigorous estimation of the success probability for Type \Rmnum{1} and Type \Rmnum{2} dual-zone hard-core processes, offering practical insights into their performance relative to standard PPP models. 
For Type \Rmnum{1} and Type \Rmnum{2} of the dual-zone hard-core processes, we get the following theorem. 

\begin{thm}
\label{thm:succ_prob}
The success probability for a receiver in the Type \Rmnum{1} and Type \Rmnum{2} dual-zone hard-core processes is given by
\begin{align}
&P_{\rm DZHCP}(T) = \nonumber \\
&\left(1 + \left(\frac{T}{G_{\rm DZHCP}}\right)^{\frac{2}{\alpha}} \int_{\left(\frac{T}{G_{\rm DZHCP}}\right)^{-\frac{2}{\alpha}}}^\infty \frac{1}{1 + t^{\frac{\alpha}{2}}} \, dt\right)^{-1},
\end{align}
where \(G_{\rm DZHCP}\) is the asymptotic gain calculated as outlined in Lemma \ref{lem:asympgain} and defined in (\ref{eqn:GCTP}).
\end{thm}

For a path-loss exponent \(\alpha = 4\), a more specific form of the approximated success probability can be found for the same processes:
\begin{cor}
\label{cor:succ_prob_alpha4}
The success probability for a receiver in the Type \Rmnum{1} and Type \Rmnum{2} dual-zone hard-core processes, with \(\alpha = 4\), simplifies to
\begin{equation}
P_{\rm DZHCP}(T) \approx \left(1 + \sqrt{\frac{T}{G_{\rm DZHCP}}} \arctan \left(\sqrt{\frac{T}{G_{\rm DZHCP}}}\right)\right)^{-1},
\end{equation}
where \(G_{\rm DZHCP}\) is defined in (\ref{eqn:GCTP}).
\end{cor}

These formulas model the success probability by adapting the classical success probability formula for a PPP, modified by the asymptotic gain which accounts for the reduced interference due to the dual-zone hard-core process. This adaptation is crucial in networks where space-reuse strategies, such as those deployed in Type \Rmnum{1} and Type \Rmnum{2} processes, significantly affect the distribution and intensity of interference. The success probability expression provides a quantifiable measure of how these strategies improve network performance relative to a standard PPP, particularly under various environmental and operational conditions dictated by the parameter \(\alpha\).

\begin{figure}
    \centering
    \includegraphics[width=0.5\textwidth]{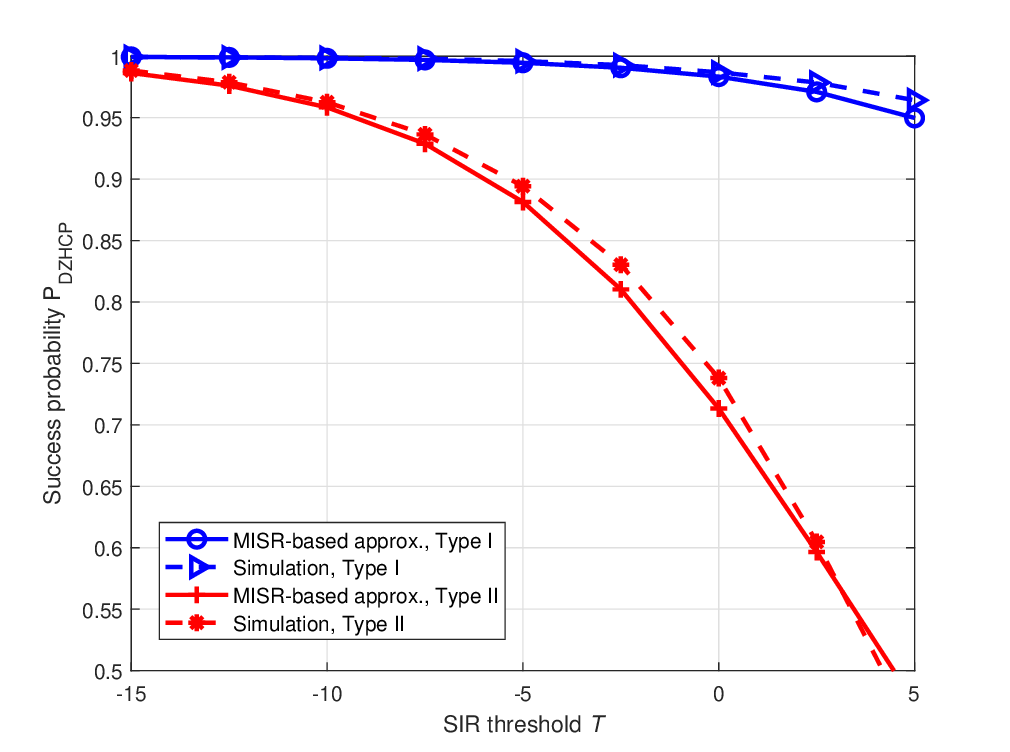}
    \caption{SIR ccdf for the dual-zone hard-core process and MISR-based approximation for \(\lambda_p=1\times10^{-4} {\rm m}^{-2}\).}
    \label{fig:highdensity}
\end{figure}

\begin{figure}
    \centering
    \includegraphics[width=0.5\textwidth]{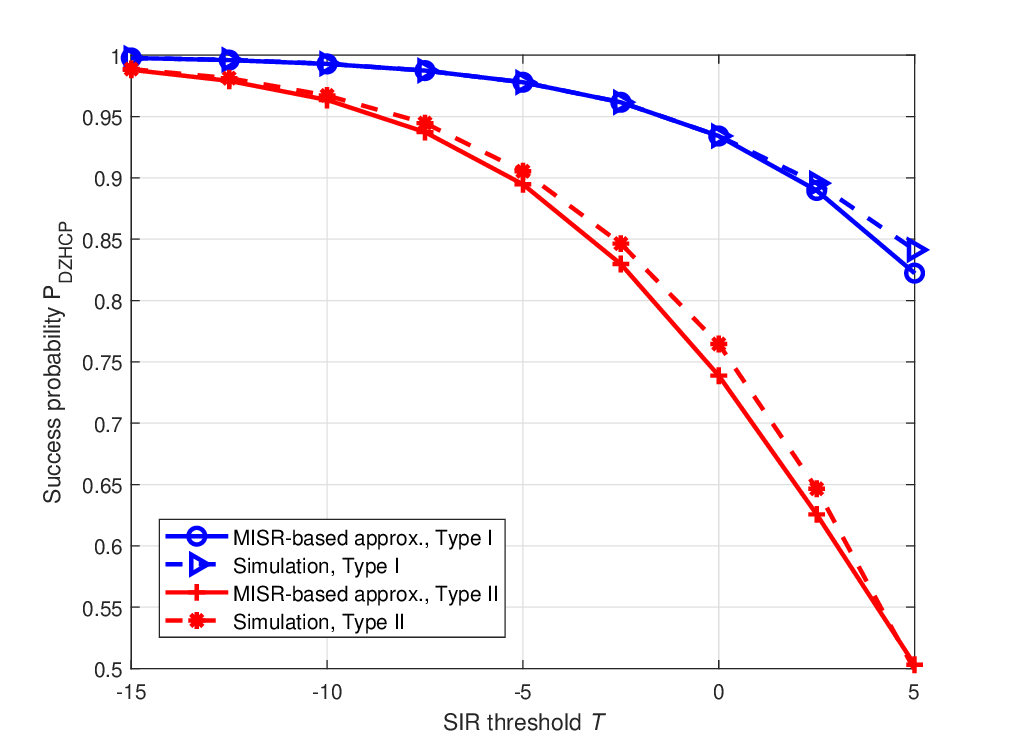}
    \caption{SIR ccdf for the dual-zone hard-core process and MISR-based approximation for \(\lambda_p=5\times10^{-5} {\rm m}^{-2}\).}
    \label{fig:lowdensity}
\end{figure}

To assess the accuracy of the MISR approximation for the bipolar case, we compare its results with simulation data in this work. Figures \ref{fig:highdensity} and \ref{fig:lowdensity} present the success probability under varying densities and for commonly used values of \(T\), showing that the proposed MISR approximation yields accurate results for all practical success probabilities. These results validate the applicability of the MISR approximation for the bipolar dual hard-core model, despite the differences in interference characteristics when compared to cellular networks.

\section{Numerical Results}
\label{sec:numerical}
The numerical results are obtained from the analytical results we have derived.
The default configurations of system model are as follows (also see Table \ref{table:parameter}).
Unless otherwise specified, the intensity of potential transmitters is set as $\lambda_p=1\times10^{-5}\mathrm{m}^{-2}$ and the virtual carrier sensing radius is $R_{\mathrm{tx}}=100$m.
The physical carrier sensing radius is always set as $R_{\mathrm{cs}}=1.2R_{\mathrm{tx}}$ and the distance between transmitter and receiver is set as $d=0.8R_{\mathrm{tx}}$.
The transmit power is set as $P_t=20$dBm (0.1W), and the path loss model is $l(r) = A r^{-\alpha}$ with $\alpha=3.5$ and $A=0.01$.
In the following figures, the legend entries `Type \Rmnum{1}' and `Type \Rmnum{2}' represent the RTS/CTS handshake protocol as modeled by the Type \Rmnum{1} and Type \Rmnum{2} dual-zone hard-core processes, respectively. Similarly, the entries `CSMA \Rmnum{1}' and `CSMA \Rmnum{2}' correspond to the CSMA protocol, modeled using the Type \Rmnum{1} and Type \Rmnum{2} Matérn hard-core processes, respectively.

\begin{table}
\centering
\caption{SYSTEM PARAMETERS}
\label{table:parameter}
\begin{tabular}[!ht]{c|c|c}
\hline
\hline
\textbf{Symbol} & \textbf{Description} & \textbf{Value}\\
\hline
$\lambda_p$ & Intensity of potential transmitters & $1\times10^{-5}\mathrm{m}^{-2}$\\
\hline
$\lambda$ & Intensity of the Type \Rmnum{1} or \Rmnum{2} process & N/A\\
\hline
$R_{\mathrm{tx}}$ & Virtual carrier sensing radius & $100$m\\
\hline
$R_{\mathrm{cs}}$ & Physical carrier sensing radius & $1.2R_{\mathrm{tx}}$\\
\hline
$d$ & Distance between transmitter and receiver & $0.8R_{\mathrm{tx}}$\\
\hline
$P_t$ & Transmit power & $20$dBm\\
\hline
$T$ & SIR threshold & 0 dB\\
\hline
$\alpha$ & Path loss exponent & $3.5$\\
\hline
\hline
\end{tabular}
\end{table}

\begin{figure}
    \centering
    \includegraphics[width=0.5\textwidth]{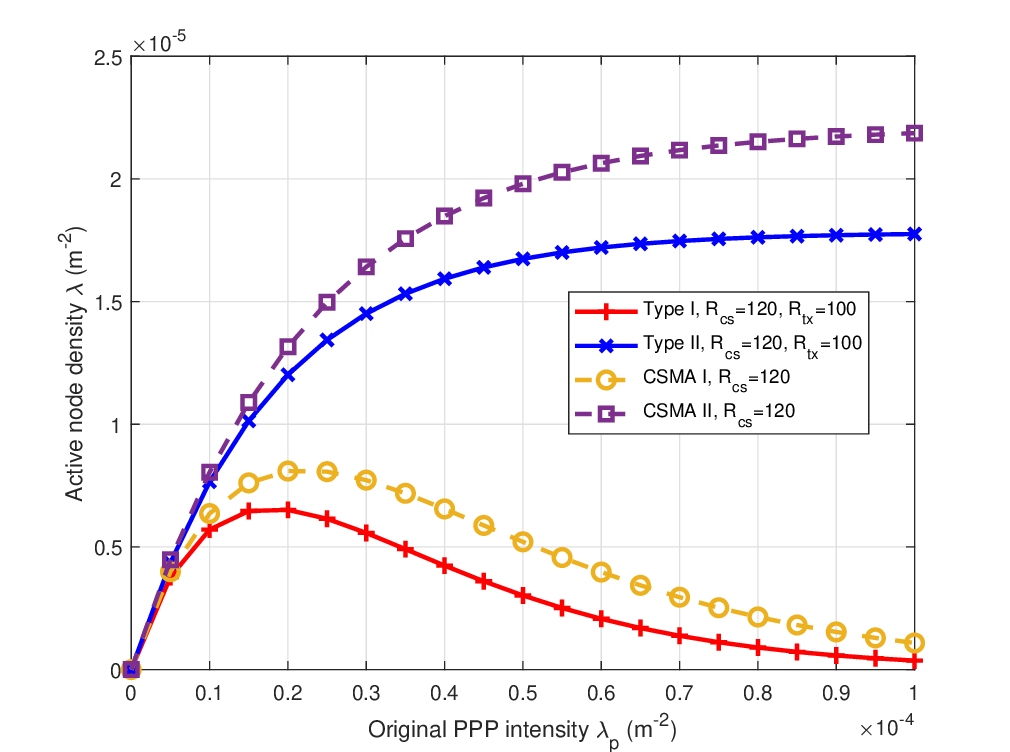}
    \caption{Active node density $\lambda$ as a function of the original PPP intensity $\lambda_p$ for various RTS/CTS configurations.}
    \label{fig:active_density_vs_PPP}
\end{figure}

Figure \ref{fig:active_density_vs_PPP} illustrates the relationship between the active node density $\lambda$ and the original PPP intensity $\lambda_p$ under different configurations of the RTS/CTS mechanism, with fixed virtual carrier sensing range $R_{\text{tx}}$ and physical carrier sensing range $R_{\text{cs}}$.
For both CSMA \Rmnum{1} and Type \Rmnum{1} models, we observe an initial increase in active node density $\lambda$ as the original PPP intensity $\lambda_p$ increases, followed by a subsequent decrease. This trend reaches its peak when $\lambda_p = \frac{1}{V_o}$, where $V_o$ represents the area of the exclusion region, yielding a maximum value of $\frac{1}{V_o} e^{-1}$. This behavior reflects the impact of the protection region, where the density of active nodes is initially enhanced by a higher node availability but later constrained as interference effects intensify.
In contrast, for the CSMA \Rmnum{2} and Type \Rmnum{2} models, active node density $\lambda$ increases continuously with the original PPP intensity $\lambda_p$, eventually approaching the reciprocal of the exclusion region. This gradual increase indicates that these models allow for higher active node density even at larger values of $\lambda_p$, suggesting a more effective management of interference as $\lambda_p$ grows.

\begin{figure}
    \centering
    \includegraphics[width=0.5\textwidth]{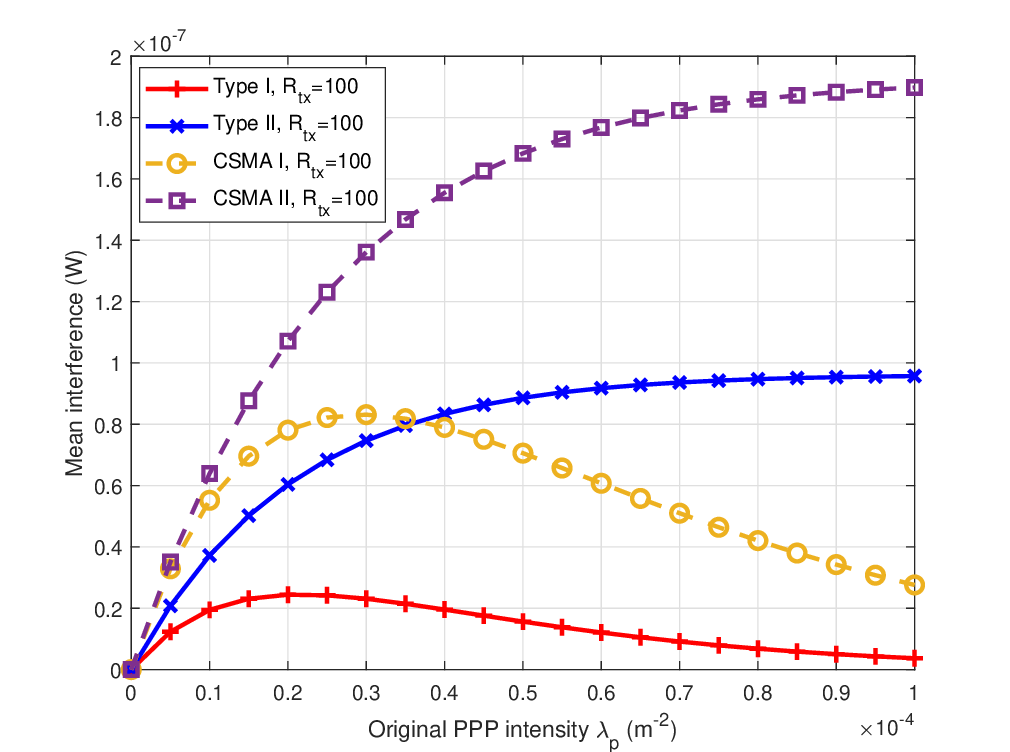}
    \caption{Mean interference as a function of the original PPP intensity \(\lambda_p\) under different schemes.}
    \label{fig:avg_interference_vs_lambda}
\end{figure}

Figure \ref{fig:avg_interference_vs_lambda} shows the variation in mean interference as the original PPP intensity \(\lambda_p\) changes, with a fixed virtual carrier sensing range $R_{\text{tx}}$, physical carrier sensing radius \(R_{\mathrm{cs}}\), and distance \(d\) between transceivers. For CSMA I and Type \Rmnum{1} schemes, the mean interference initially increases with \(\lambda_p\) before eventually decreasing. This trend mirrors the behavior observed in their active node densities, with both metrics reaching their peak at the same \(\lambda_p\) value. Similarly, for CSMA II and Type \Rmnum{2} schemes, the mean interference exhibits a similar trend with respect to \(\lambda_p\), matching the behavior of their active node densities as well. 
However, it is notable that when \(\lambda_p < 3 \times 10^{-5}\mathrm{m}^{-2}\), the active node density for CSMA I is lower than that for Type \Rmnum{2}, yet the mean interference for CSMA I exceeds that of Type \Rmnum{2}. This indicates that at low densities, even though CSMA I has fewer active nodes than Type \Rmnum{2}, the structural differences in their protection regions result in a higher likelihood of interference in CSMA I. Consequently, the overall mean interference remains higher for CSMA I in this regime.

\begin{figure}
    \centering
    \includegraphics[width=0.49\textwidth]{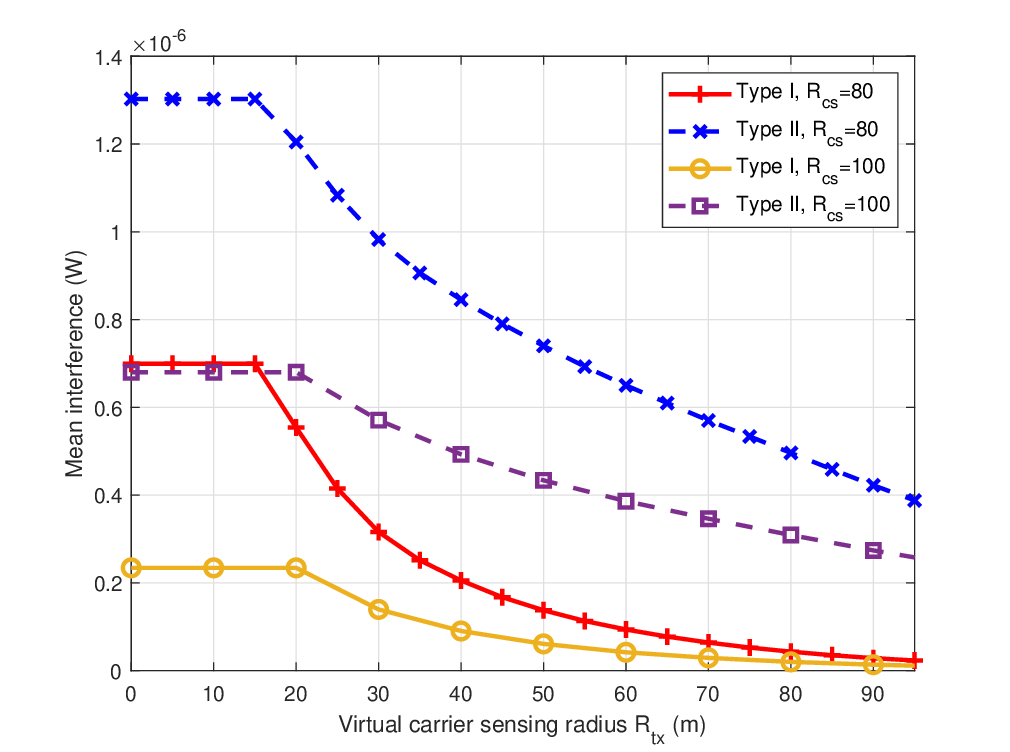}
    \caption{Mean interference as a function of virtual carrier sensing radius \( R_{\mathrm{tx}} \) with fixed physical carrier sensing radius \( R_{\mathrm{cs}} \) and distance \( d \).}
    \label{fig:avg_interference_vs_Rtx}
\end{figure}

Figure \ref{fig:avg_interference_vs_Rtx} illustrates the relationship between mean interference and the virtual carrier sensing range \( R_{\mathrm{tx}} \) under varying schemes with a constant physical carrier sensing radius \( R_{\mathrm{cs}} \) and distance \( d \). When \( R_{\mathrm{tx}} \) is relatively small, the mean interference does not significantly change with increasing \( R_{\mathrm{tx}} \). This is because, at smaller values of \( R_{\mathrm{tx}} \), the protection region of the transmitter fully covers the protection region of the receiver, leading to an interference level that decreases as \( R_{\mathrm{cs}} \) increases.
However, as \( R_{\mathrm{tx}} \) continues to increase and the transmitter's protection region no longer completely encompasses the receiver's protection region, the mean interference starts to decrease with increasing \( R_{\mathrm{tx}} \). 

\begin{figure}
    \centering
    \includegraphics[width=0.5\textwidth]{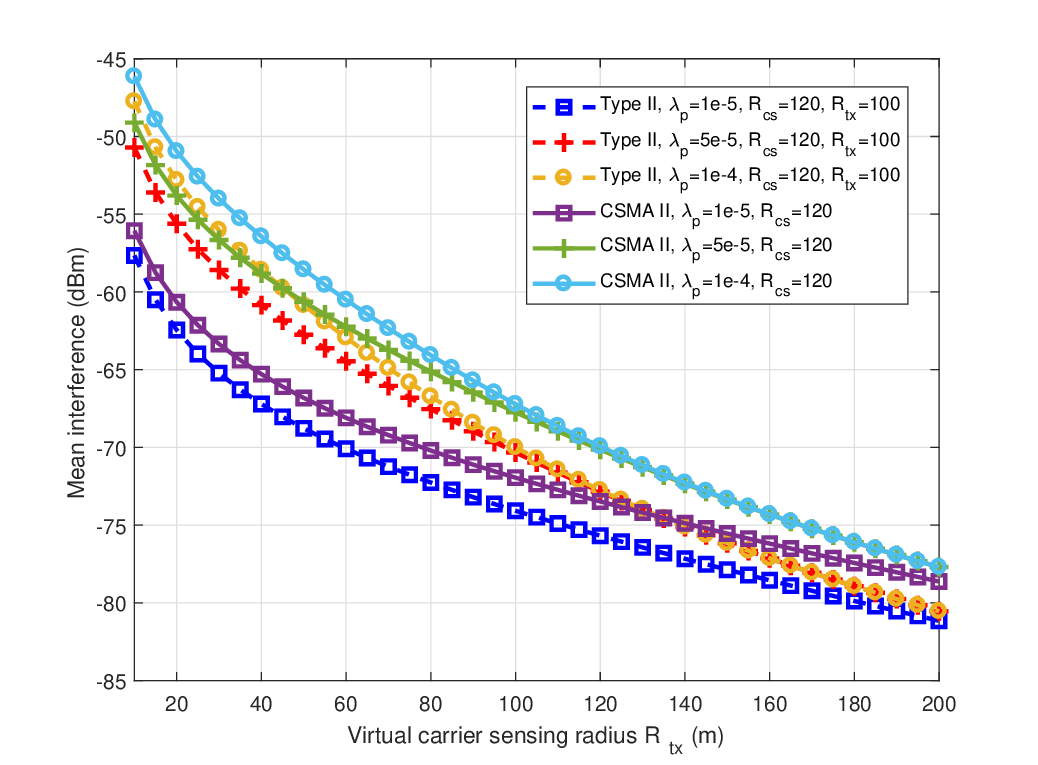}
    \caption{Mean interference as a function of virtual carrier sensing radius \( R_{\mathrm{tx}} \) for Type \Rmnum{2} and CSMA II models with physical carrier sensing radius \( R_{\mathrm{cs}} = 1.2 R_{\mathrm{tx}} \).}
    \label{fig:avg_interference_vs_Rtx_TypeII_CSMAII}
\end{figure}

Figure \ref{fig:avg_interference_vs_Rtx_TypeII_CSMAII} examines the relationship between mean interference and virtual carrier sensing range \( R_{\mathrm{tx}} \) in Type \Rmnum{2} and CSMA II models, with a virtual carrier sensing radius \( R_{\mathrm{cs}} = 1.2 R_{\mathrm{tx}} \). In both models, as \( R_{\mathrm{tx}} \) increases, indicating an expansion in protection region, the three CSMA II curves begin to converge, showing a similar trend. A similar convergence phenomenon is observed for the three curves in the Type \Rmnum{2} model. 
This pattern suggests that for the Type \Rmnum{2} hard-core point processes, as the protection region expands, the mean interference level becomes more dependent on the size of the protection region rather than the original PPP intensity. 

\begin{figure}
    \centering
    \includegraphics[width=0.5\textwidth]{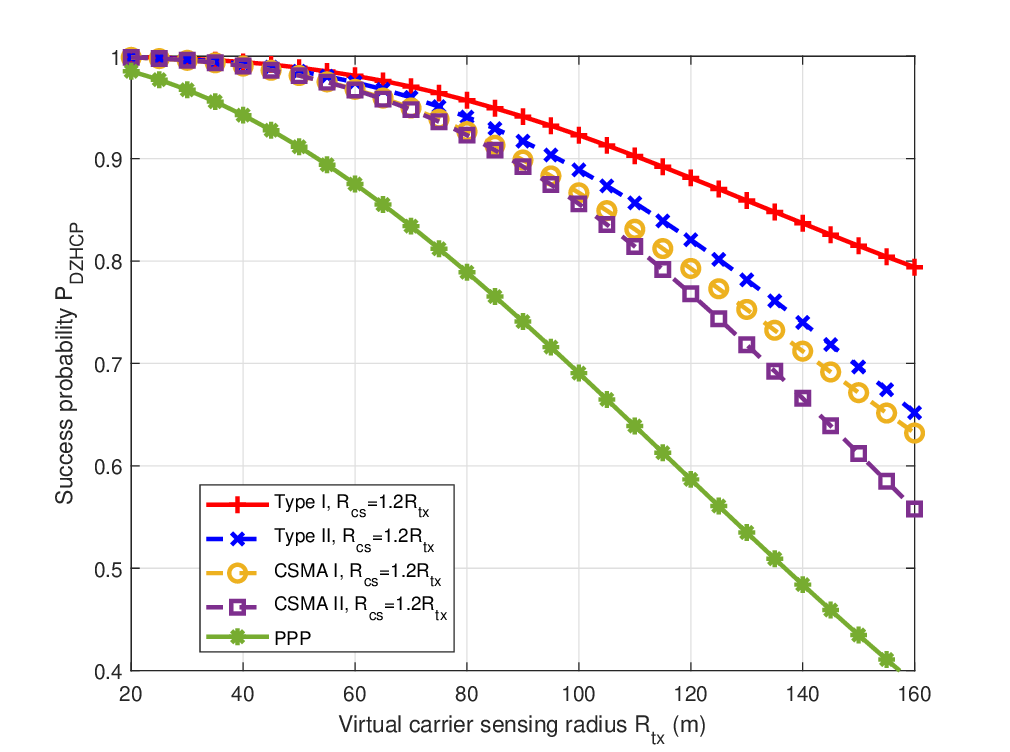}
    \caption{Success probability as a function of virtual carrier sensing radius \( R_{\mathrm{tx}} \) with fixed ratio \( d : R_{\mathrm{tx}} : R_{\mathrm{cs}} = 0.8:1:1.2 \).}
    \label{fig:success_probability_vs_Rtx1}
\end{figure}

Figure \ref{fig:success_probability_vs_Rtx1} illustrates the relationship between success probability \( P_{\rm DZHCP} \) and the virtual carrier sensing range \( R_{\mathrm{tx}} \), under the fixed ratio \( d : R_{\mathrm{tx}} : R_{\mathrm{cs}} = 0.8 : 1 : 1.2 \). As \( R_{\mathrm{tx}} \) increases, there is a noticeable decrease in the success probability \( P_{\rm DZHCP} \). This trend can be explained by two competing effects of increasing \( R_{\mathrm{tx}} \): while a larger \( R_{\mathrm{tx}} \) helps to reduce mean interference by expanding the protection region, it simultaneously weakens the received signal strength, which negatively impacts the success probability. This phenomenon highlights that, under proportional scaling, the rate of reduction in mean interference is slower than the rate of decrease in received signal strength, leading to an overall reduction in success probability as \( R_{\mathrm{tx}} \) grows.

\begin{figure}
    \centering
    \includegraphics[width=0.5\textwidth]{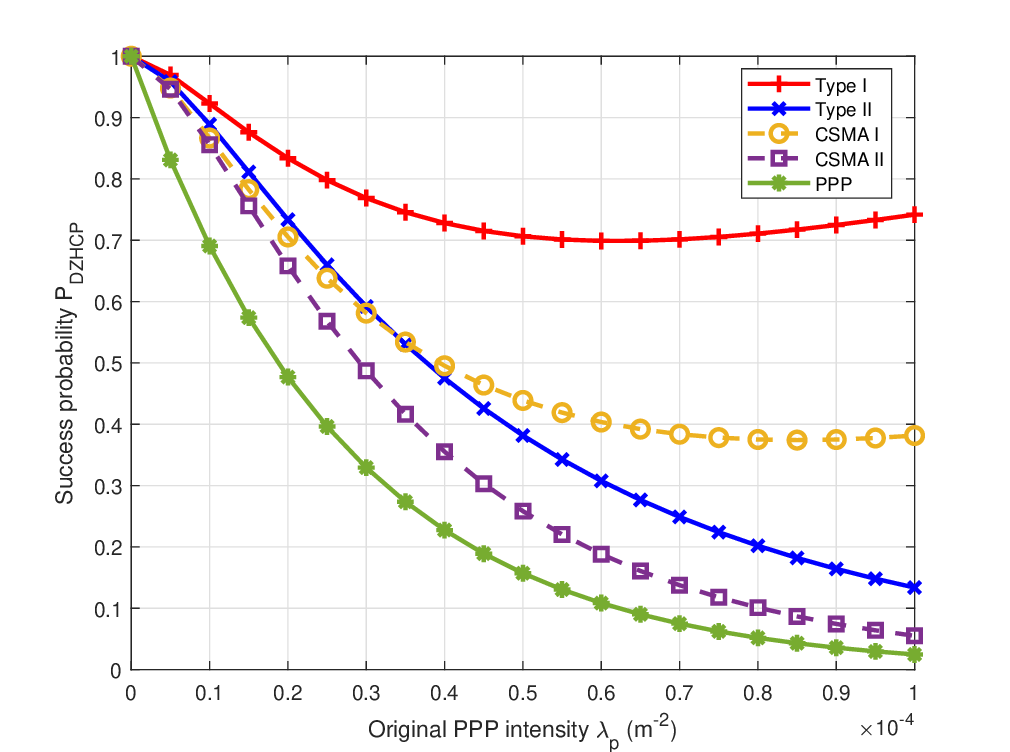}
    \caption{Success probability \( P_{\rm DZHCP} \) as a function of PPP intensity \( \lambda_p \).}
    \label{fig:success_probability_vs_lambda_p1}
\end{figure}

Figure \ref{fig:success_probability_vs_lambda_p1} shows the relationship between the success probability \( P_{\rm DZHCP} \) and the original PPP intensity \( \lambda_p \) for different network models. For the Type \Rmnum{1} and CSMA I models, the success probability \( P_{\rm DZHCP} \) initially decreases with an increase in \( \lambda_p \), then increases after reaching a certain point. This behavior can be explained by the effect of mean interference, which first rises with increasing \( \lambda_p \) but then diminishes as the density becomes very high.
On the other hand, for the Type \Rmnum{2} and CSMA II models, the success probability \( P_{\rm DZHCP} \) generally decreases as \( \lambda_p \) increases, but this decrease is not linear. Specifically, at lower values of \( \lambda_p \), the success probability may drop quickly, while at higher values, the rate of decrease slows down, showing a gradual decline. This pattern suggests that for these models, the success probability asymptotically stabilizes as the PPP intensity increases.

\begin{figure}
    \centering
    \includegraphics[width=0.5\textwidth]{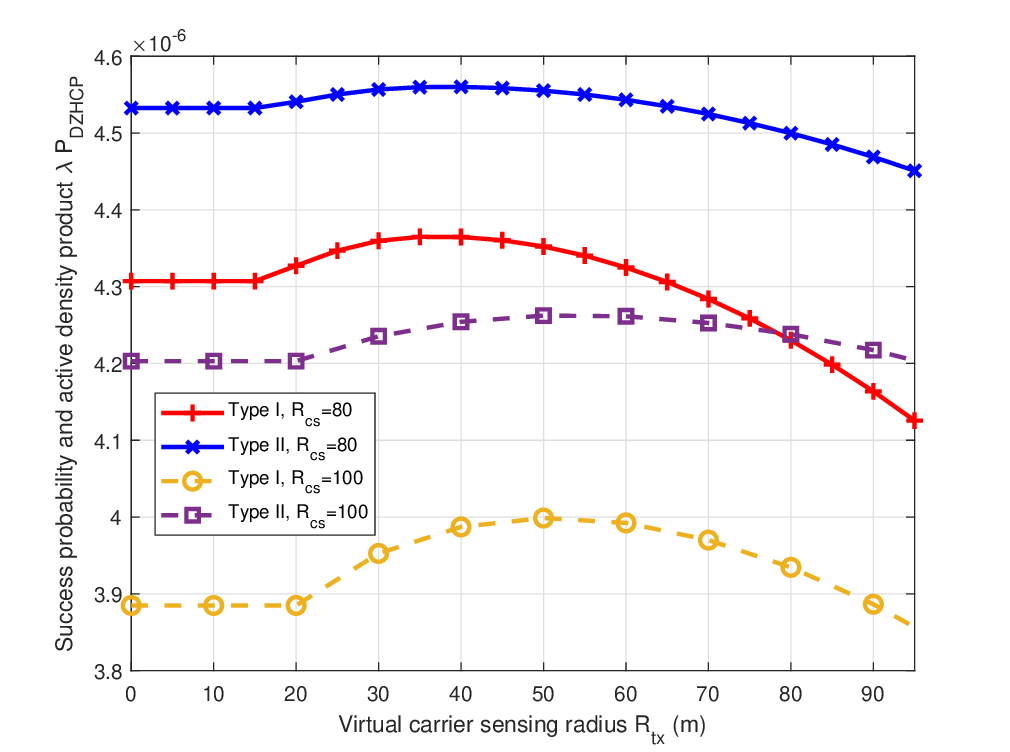}
    \caption{Success probability and active density product as a function of \( R_{\rm tx} \) (fixed \( R_{\rm cs} \) and \( d \)).}
    \label{fig:success_probability_vs_Rtx}
\end{figure}

Figure \ref{fig:success_probability_vs_Rtx} illustrates the relationship between the product of success probability and active density with respect to the virtual carrier sensing range \( R_{\mathrm{tx}} \), under fixed values of \( R_{\rm cs} \) and \( d \). 
When \( R_{\rm tx} \) is small, this metric remains relatively constant as \( R_{\rm tx} \) increases. This stability occurs because, at lower \( R_{\rm tx} \) values, the protection region of the transmitter completely covers the receiver’s protection region, leading to no interference fluctuation. However, as \( R_{\rm tx} \) continues to increase, there comes a point where the transmitter's protection region no longer fully covers the receiver's region. Beyond this point, the metric first increases with \( R_{\rm tx} \) and then decreases.
Notably, the maximum value for both cases (\( R_{\rm cs} = 80 \)m and \( R_{\rm cs} = 100 \)m) occurs around \( R_{\rm tx} = 0.5 R_{\rm cs} \), indicating that the network throughput is optimal at this proportion. This finding highlights the importance of maintaining a balanced ratio between \( R_{\rm tx} \) and \( R_{\rm cs} \) to achieve optimal performance.

\begin{figure}
    \centering
    \includegraphics[width=0.5\textwidth]{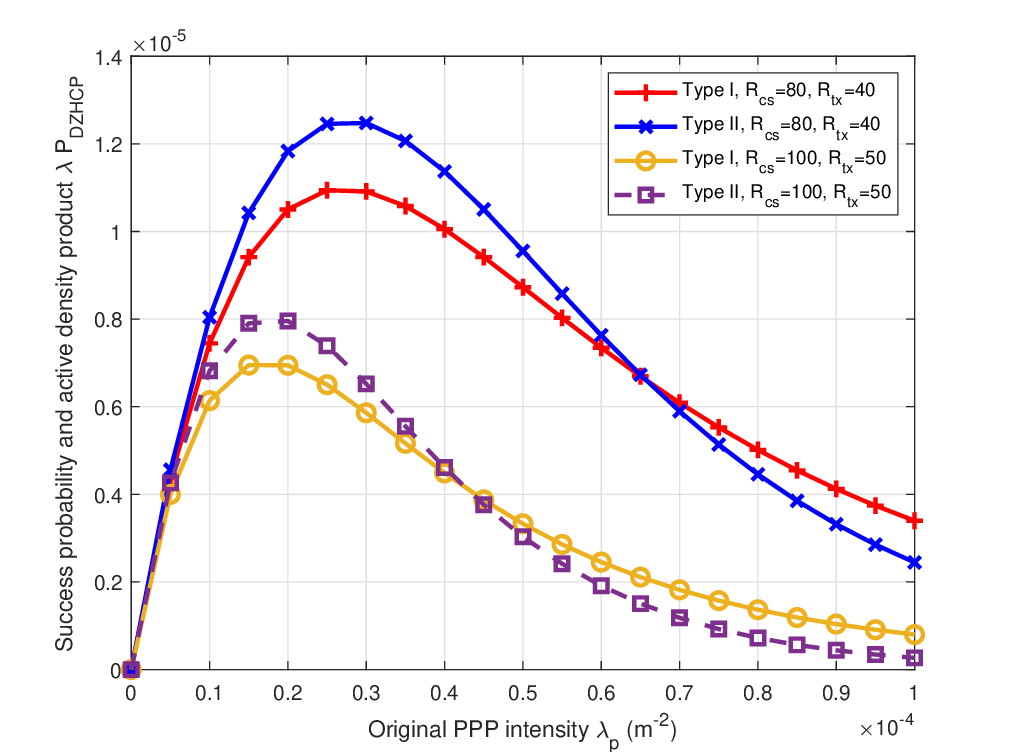}
    \caption{Success probability and active density product as a function of the original PPP intensity \( \lambda_p \).}
    \label{fig:success_probability_vs_lambda_p}
\end{figure}

Figure \ref{fig:success_probability_vs_lambda_p} illustrates the relationship between the product of success probability and active density with respect to the intensity of the original PPP \( \lambda_p \).
As shown in the figure, both Type \Rmnum{1} and Type \Rmnum{2} networks exhibit a trend where this metric initially increases with the increase in \( \lambda_p \), reaching a peak, and then decreases as \( \lambda_p \) continues to grow. This pattern suggests an optimal network density where the product of success probability and active density is maximized.
For both Type \Rmnum{1} and Type \Rmnum{2}, under fixed values of \( R_{\rm cs} \) and \( R_{\rm tx} \), the maximum values of this metric appear at similar values of \( \lambda_p \). This indicates that, when \( R_{\rm tx} \) and \( R_{\rm cs} \) are held constant, the optimal operating point in terms of network density is comparable for both types. This finding highlights a similar balance point in network density for achieving peak performance in both Type \Rmnum{1} and Type \Rmnum{2} networks.

\section{Conclusion}
\label{sec:conclusion}
In this paper, we have developed marked point process models to delineate the spatial distribution of transceivers under the RTS/CTS handshake mechanism in WLAN environments. These models have enabled a quantitative examination of how mean interference and the success probability of transmissions are influenced by critical system parameters such as carrier sensing ranges and the density of network nodes. Our analysis not only sheds light on the dynamics of mean interference but also elaborates on the conditions that enhance the probability of successful transmissions, thus providing a holistic view of network performance under various spatial configurations.

The findings underscore the significant impact of the spatial arrangement of transceivers and the implementation of specific handshake mechanisms on network efficiency, particularly concerning interference management and transmission success. This research highlights the utility of advanced point process modeling techniques in more complex scenarios, suggesting that such methods could be pivotal in designing future wireless networks that optimize interference suppression while enhancing coverage and connectivity. Looking forward, we envision extending these models to accommodate dynamic network environments and real-world operational conditions, which could further refine our understanding and application of these methodologies in the planning and optimization of next-generation wireless systems.

\appendices

\section{Proof of Theorem 1}
\label{appendix:A}

We define $\mathbb{\widetilde{M}}$ as the set of all marked point processes. The interference experienced by the receiver located at $z_o$ is calculated by summing the powers received from all other transmitters in the network, where each transmitter is specified by its location $x$ and orientation $\theta$ within the point process $\widetilde{\Phi}$.

The expectation of interference, $E_{(o,\theta_o)}^!(I_{z_o})$, under the reduced Palm probability measure $P_{(o,\theta_o)}^!$, conditioned on a transmitter at the origin $o$ with orientation $\theta_o$, is:
\begin{align}
E_{(o,\theta_o)}^!&(I_{z_o}) = E_{(o,\theta_o)}^!\Big(\sum_{(x,\theta)\in\widetilde{\Phi}}P_t l(|x-z_o|)\Big) \nonumber \\
&= \lambda P_t \int_{\mathbb{R}^2 \times [0, 2\pi]} l(|x - z_o|) \mathcal{K}_{\theta_o}(\mathrm{d}(x, \theta)). \label{equ:mean_I}
\end{align}

The reduced second-order factorial measure $\mathcal{K}_{\theta_o}(B \times L)$ for regions $B \subset \mathbb{R}^2$ and $L \subset [0, 2\pi]$ is defined as
\begin{eqnarray}
\mathcal{K}_{\theta_o}(B \times L) = \frac{1}{\lambda} \int_{\mathbb{\widetilde{M}}} \widetilde{\varphi}(B \times L) P_{(o, \theta_o)}^!(\mathrm{d}\widetilde{\varphi}),
\end{eqnarray}
where $\widetilde{\varphi}(B \times L)$ counts the points in $\widetilde{\Phi}$ that fall within $B$ and $L$. This measure represents the expected number of points in $B \times L$ under the Palm distribution $P_{(o, \theta_o)}^!$.

For the second-order factorial moment measure $\alpha^{(2)}(B_1 \times B_2 \times L_1 \times L_2)$, which represents pairs of points in specified spatial and mark regions, we have:
\begin{align}
\alpha^{(2)}(B_1 &\times B_2 \times L_1 \times L_2) = \nonumber \\
&\frac{\lambda^2}{2\pi} \int_{B_1 \times L_1} \mathcal{K}_{\theta}((B_2 - x) \times L_2) \, \mathrm{d}(x, \theta). \label{equ:alpha2}
\end{align}

The second-order product density, $\varrho^{(2)}(x_1, x_2, \theta_1, \theta_2)$, describes the likelihood of finding pairs of points with specific relative positions and marks. For a stationary process, we have:
\begin{align}
&\alpha^{(2)}(B_1 \times B_2 \times L_1 \times L_2) = \nonumber \int_{B_1 \times L_1} \nonumber\\
&\left( \int_{(B_2 - x_1) \times L_2} \varrho^{(2)}(x_2, \theta_1, \theta_2) \, \mathrm{d}(x_2, \theta_2) \right) \mathrm{d}(x_1, \theta_1). \label{equ:alpha3}
\end{align}

This leads to:
\begin{eqnarray}
\mathcal{K}_{\theta_o}(B \times L) = \frac{2\pi}{\lambda^2} \int_{B \times L} \varrho^{(2)}(x, \theta_o, \theta) \, \mathrm{d}(x, \theta). \label{equ:derive}
\end{eqnarray}

Thus, the mean interference $E_{(o, \theta_o)}^!(I_{z_o})$ becomes
\begin{align}
E_{(o, \theta_o)}^!&(I_{z_o}) = \nonumber \\
\frac{2\pi P_t}{\lambda} &\int_{\mathbb{R}^2 \times [0, 2\pi]} l(|x - z_o|) \varrho^{(2)}(x, \theta_o, \theta) \, \mathrm{d}(x, \theta). \label{equ:mean_I_2}
\end{align}
This expression incorporates the dependencies in location and orientation, providing insight into interference distribution.

To derive the second-order product density $\varrho^{(2)}(x,\theta_o,\theta)$, we utilize its relation to the two-point Palm probability $k(x,\theta_o,\theta)$. This probability quantifies the likelihood that two transmitters, separated by a vector $x$ and with orientation marks $\theta_o$ and $\theta$, are both active (i.e., not thinned out by the RTS/CTS mechanism). The relationship is given by the following equation derived from stochastic geometry principles:
\begin{eqnarray}
\varrho^{(2)}(x,\theta_o,\theta)=\frac{\lambda_p^2}{4\pi^2}k(x,\theta_o,\theta).
\end{eqnarray}
Here, $k(x,\theta_o,\theta)$ acts as a modulation factor adjusting this density based on the spatial and orientational configuration of transmitters.

We consider a typical scenario where the orientation of the initial transmitter is normalized to zero, $\theta_o = 0$, simplifying the analysis without loss of generality. The typical receiver is hence positioned at $z_o = (d,0)$. The interference experienced by this receiver, $I_{z_o}$, is then calculated as:
\begin{equation}
E_{(o,0)}^!(I_{z_o}) = \frac{\lambda_p^2P_t}{2\pi\lambda}\int_{\mathbb{R}^2\times[0,2\pi]}l(|x-z_o|)k(x,0,\theta)\mathrm{d}(x,\theta),
\end{equation}
where $l(|x-z_o|)$ denotes the path-loss function, and $k(x,0,\theta)$ modifies the intensity of points contributing to the interference based on their spatial and mark relationships.

This integral can be expressed in polar coordinates as
\begin{align}
E_{(o,0)}^!&(I_{z_o}) = \frac{\lambda_p^2P_t}{2\pi\lambda}\int_0^\infty\int_0^{2\pi}\int_0^{2\pi} \nonumber \\
&l(\sqrt{r^2-2rd\cos\beta+d^2})k(r,\beta,0,\theta)r\mathrm{d}\theta\mathrm{d}\beta\mathrm{d}r. \label{equ:inter_final_proof}
\end{align}
Here, $r$ is the radial distance from the origin, $\beta$ is the angle difference from the direction to the receiver, and $\theta$ is the mark of the interfering transmitter.

The function $k(r,\beta,0,\theta)$, critical for this calculation, reflects the conditional probability that two transmitters, separated by distance $r$ with phase angle difference $\beta$ and marks $0$ and $\theta$, do not violate the RTS/CTS protocol rules. Specifically, it evaluates to zero if the triple $(r,\beta,\theta)$ falls within any of the sets $S_1$, $S_2$, or $S_3$, indicating immediate exclusion due to proximity within the guard regions. Otherwise, the transmitters are retained with a probability of $\exp(-\lambda_p V(r, \beta, \theta))$, reflecting the void probability of the exclusion region $V(r, \beta, \theta)$ for the PPP.

\section{Proof of Theorem 2}
\label{appendix:B}
The derivation of (\ref{equ:inter_final2}) follows the same line as that of the Type \Rmnum{1} dual-zone hard-core process in Theorem \ref{thm:inter_1} except for the difference in the probability that both of the transceiver pairs are retained, namely, $k(r,\beta,0,\theta)$, which is the probability that two transceiver
pairs are both retained, is different from the Type \Rmnum{1} dual-zone hard-core process. since the time stamp marks are taken into consideration.
We should consider not only the geometric relationship (see Figure \ref{fig:twopair}), but also the relationship between the time stamp marks of the two transceiver pairs.
When $(r,\beta,\theta) \in S_1\bigcup (S_2\bigcap S_3)$, at least one of the considered two transceiver pairs has to be removed after comparing their time stamp marks and thus $k(r,\beta,0,\theta)$ is zero.
When $(r,\beta,\theta)\in\overline{S_1}\bigcap\overline{S_2}\bigcap\overline{S_3}$, there is no direct comparison between the time stamp marks of the considered two transceiver pairs. Let their time stamp marks be $t_1$ and $t_2$ respectively.
So by separately considering the cases of $t_1 \geq t_2$ and $t_1 < t_2$, we get
\begin{align}
\label{eqn:k-two-conditions}
k(r,\beta,0,\theta)=& \int_0^1 e^{-\lambda_p t_1 V_o} \left[\int_0^{t_1} e^{-\lambda_p t_2 (V - V_o)} \mathrm{d}t_2\right] \mathrm{d}t_1\nonumber\\
&+ \int_0^1 e^{-\lambda_p t_2 V_o} \left[\int_0^{t_2} e^{-\lambda_p t_1 (V - V_o)} \mathrm{d}t_1\right] \mathrm{d}t_2\nonumber\\
= 2 \int_0^1 &e^{-\lambda_p t_1 V_o} \left[\int_0^{t_1} e^{-\lambda_p t_2 (V - V_o)} \mathrm{d}t_2\right] \mathrm{d}t_1,
\end{align}
which can be evaluated as $2\eta(V)$. Here, $\int_0^{t_1} e^{-\lambda_p t_2 (V - V_o)} \mathrm{d}t_2$ is the conditional probability that, when the first transceiver pair is marked by $t_1$, there is no other transceiver pair of mark smaller than $t_2 < t_1$ lying in the region covered by the second transceiver pair only; $\int_0^{t_2} e^{-\lambda_p t_1 (V - V_o)} \mathrm{d}t_1$ is analogously interpreted. When either $(r,\beta,\theta)\in\overline{S_1}\bigcap{S_2}\bigcap\overline{S_3}$ or $(r,\beta,\theta)\in\overline{S_1}\bigcap\overline{S_2}\bigcap{S_3}$, exactly one of the considered transmitter is within the RTS/CTS cleaned region of the other, and hence only one of the terms in (\ref{eqn:k-two-conditions}) should be taken into account, resulting into $k(r,\beta,0,\theta) = \eta(V)$.

\bibliographystyle{IEEEtran}
\bibliography{csma-isit}

\end{document}